\documentclass[submission,copyright,creativecommons,hidelinks]{eptcs}
\usepackage[utf8]{inputenc}
\usepackage{proof}
\usepackage{amsmath,amssymb, color, tikz, amsthm}

\usepackage[pdf]{pstricks}
\usepackage[ruled,vlined,linesnumbered]{algorithm2e}
\SetKwInOut{Input}{Input}
\SetKwInOut{Precondition}{Precondition}

\usepackage{cleveref}
\usepackage{stmaryrd}
\usepackage{booktabs}
\usepackage{float}
\usepackage{thmtools} 
\usepackage{thm-restate}

\usetikzlibrary{arrows.meta, shapes, positioning, calc, quotes}

\newcommand{\powerset}[1]{\mathcal{P}(#1)}

\newcommand\ldiaarg[1]{\langle#1\rangle}

\newcommand{\LL}{\mathcal{L}}

\newcommand{\cP}{\mathcal{P}}

\newcommand{\hatK}{\widehat{K}}

\newcommand{\PAL}{\ensuremath{\mathsf{PAL}}\xspace}
\newcommand{\DEL}{\ensuremath{\mathsf{DEL}}\xspace}
\newcommand{\EL}{\ensuremath{\mathsf{EL}}\xspace}

\newcommand{\Tr}{\mathsf{true}}
\newcommand{\Fa}{\mathsf{false}}
\newcommand{\tr}{\mathsf{tr}}

\newcommand{\PSPACE}{\mathsf{PSPACE}}

\newcommand{\union}{\cup}

\newcommand{\outof}{\mathsf{of}}
\newcommand{\goto}{\mathsf{goto}}
\newcommand{\change}{\mathsf{change}}

\DeclareMathOperator*{\bigsem}{;}

\newtheorem{theorem}{Theorem}[section]

\newtheorem{corollary}[theorem]{Corollary}

\newtheorem{lemma}[theorem]{Lemma}

\newtheorem{definition}[theorem]{Definition}

\newtheorem{example}[theorem]{Example}
\newtheorem{conjecture}[theorem]{Conjecture}

\newcommand{\twopartdef}[4]{
  \left\{
    \begin{array}{@{}ll}
      #1 & #2 \\
      #3 & #4
    \end{array}
  \right.
}

\newcommand{\F}{\ensuremath{\mathcal{F}}}
\newcommand{\X}{\ensuremath{\mathcal{X}}}

\newcommand{\chk}{\mathsf{check}}

\newcommand{\chkfcdelK}{\mathsf{checkDELK}}
\renewcommand{\phi}{\varphi}

\providecommand{\event}{TARK 2025} 

\usepackage{iftex}
\usepackage{array}

\ifpdf
  \usepackage{underscore}         
  \usepackage[T1]{fontenc}        
\else
  \usepackage{breakurl}           
\fi

\title{Comparing State-Representations for DEL Model Checking}
\author{Gregor Behnke\qquad Malvin Gattinger\qquad Haitian Wang
\institute{ILLC, University of Amsterdam\\
Amsterdam, The Netherlands}
\email{g.behnke@uva.nl\qquad malvin@w4eg.eu\qquad\qquad h.wang12@uva.nl}
\and
Avijeet Ghosh
\institute{Chennai Mathematical Institute\\
Chennai, India}
\email{avi.ghosh23@gmail.com}
}
\def\titlerunning{Comparing State-Representations for DEL Model Checking}
\def\authorrunning{G.~Behnke, M.~Gattinger, A.~Ghosh \& H.~Wang}
\begin{document}
\maketitle

\begin{abstract}
  Model checking with the standard Kripke models used in (Dynamic) Epistemic Logic leads to scalability issues.
  Hence alternative representations have been developed, in particular symbolic structures based on Binary Decision Diagrams (BDDs) and succinct models based on mental programs.
  While symbolic structures have been shown to perform well in practice, their theoretical complexity was not known so far.
  On the other hand, for succinct models model checking is known to be PSPACE-complete, but no implementations are available.

  We close this gap and directly relate the two representations.
  We show that model checking DEL on symbolic structures encoded with BDDs is also PSPACE-complete.
  In fact, already model checking Epistemic Logic without dynamics is PSPACE-complete on symbolic structures.
  We also provide direct translations between BDDs and mental programs.
  Both translations yield exponential outputs.
  For the translation from mental programs to BDDs we show that no small translation exists.
  For the other direction we conjecture the same.
\end{abstract}

\section{Introduction}

Reasoning about knowledge and its representation has become an important line of study in Artificial Intelligence, most importantly in Epistemic Planning~\cite{BelleSpecialIssue}.
A widely used logical framework to reason about the knowledge of multiple intelligent agents is Epistemic Logic, a version of Modal Logic~\cite{ModalLogicBlackburnRV01}.
Its standard semantics uses Kripke models with possible worlds to represent what agents know.
Moreover, knowledge can change.
For example in a card game, a player may announce its card to a subset of players privately.
Then the knowledge of all players changes: some players come to know the card, and other players may still observe that such a private announcement is made.
To formalise and reason about such knowledge updates,
a widely used framework is Dynamic Epistemic Logic ($\DEL$)~\cite{BMS1998:PAckPS,DELbook},
of which the simplest version is Public Announcement Logic ($\PAL$)~\cite{PALPlaza07,PALComplexityLutz06}.

One way to make $\DEL$ and $\PAL$ practically useful is \emph{model checking}:
given a model and a formula, decide whether the formula holds.
It is known that on Kripke models the computational complexity of this task for $\PAL$ is in $P$ and for $\DEL$ is in $\PSPACE$~\cite{ComplexityDELAucherS13}. Model-checkers for \PAL have been used in verifying protocols with security applications like the Russian card problem~\cite{vDvHvdMJ2006:RusCards,LanDu2017:PerfSecCom}.
Moreover, \DEL model checkers like SMCDEL~\cite{GattingerThesis2018} can be used in combination with large language models (LLMs) in order to improve the explainability and reasoning ability compared to plain LLMs~\cite{BelleTang2024:ToMLM}.

Model checking has scalability issues, because Kripke models enumerate all possible situations.
For example, every possible way to distribute cards among players needs a possible world, leading to an exponentially large model.
Different techniques have been developed to represent Kripke models more compactly, to make epistemic model-checking easier.
Here we focus on two such works for \DEL:
\begin{itemize}
\item Symbolic structures
  based on Binary Decision Diagrams (BDDs)~\cite{vBEGS2018:SymDEL,GattingerThesis2018,MG2023:ZDD-DEL-TARK}.
  This approach uses the famous data structure from \cite{Bryant86} and is similar to techniques for temporal logics~\cite{ChakiBDDbasedMC2018}.
\item Succinct models
  based on mental programs (MPs)~\cite{CharSucc2017}, a variant of PDL~\cite{FL1979:PDL}, and similar to regular languages.
  In particular these programs can encode the epistemic accessibility relations.
\end{itemize}

\begin{example}
To illustrate the different representations of Kripke models, consider the model below.
It has four worlds and two agents.
Agent A considers all $q$-worlds possible and knows nothing else, whereas Agent B only considers $p$-worlds possible and knows the value of $q$.
Next to the model we show how these relations are encoded by BDDs and by mental programs.
We note that one relation results in a relatively small BDD, while the other is more suitable to be encoded as a short mental programs.
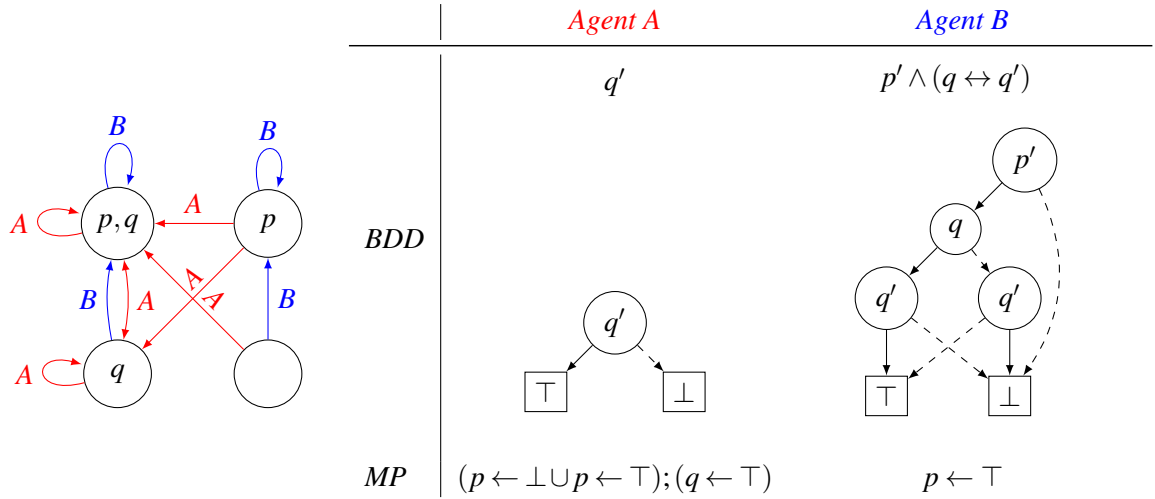
\newcommand{\aliceBDD}{
  \begin{tikzpicture}[>=latex, node distance=13mm]
    \node (p0) [draw,circle] {$q'$};
    \node (p1) [draw,rectangle, below left of=p0] {$\top$};
    \node (p2) [draw,rectangle, below right of=p0] {$\bot$};
    \draw (p0) [->,solid] -- (p1);
    \draw (p0) [->,dashed] -- (p2);
    \node (top) [above of=p0, node distance=32mm] {$q'$};
  \end{tikzpicture}}
\newcommand{\bobBDD}{
  \begin{tikzpicture}[>=latex, node distance=13mm]
    \node (p) [draw,circle] {$p'$};
    \node (q) [draw,circle,below left of=p] {$q$};
    \node (Tq2) [draw,circle,below left of=q] {$q'$};
    \node (Fq2) [draw,circle,below right of=q,xshift=-2mm] {$q'$};
    \node (T) [draw,rectangle,below of=Tq2] {$\top$};
    \node (F) [draw,rectangle,below of=Fq2] {$\bot$};
    \draw (p0) edge [->,dashed,bend left=30] (F);
    \draw (p0) [->] -- (q);
    \draw (q) [->] -- (Tq2);
    \draw (q) [->,dashed] -- (Fq2);
    \draw (Tq2) [->] -- (T);
    \draw (Tq2) [->,dashed] -- (F);
    \draw (Fq2) [->] -- (F);
    \draw (Fq2) [->,dashed] -- (T);
    \node (top) [above of=q, node distance=2cm] {$p' \land (q \leftrightarrow q')$};
  \end{tikzpicture}}
\begin{center}
\begin{tikzpicture}[>=latex, node distance=2cm, baseline=(p.south), every edge quotes/.style={fill=white,font=\footnotesize}]
  \node (pq) [draw,circle,] {$p, q$};
  \node (p) [draw,circle,minimum width=9mm,right of=pq] {$p$};
  \node (q) [draw,circle,minimum width=9mm,below of=pq] {$q$};
  \node (e) [draw,circle,minimum width=9mm,right of=q]  {$$};
  \draw (pq) edge [red, loop left] node {A} (pq);
  \draw (pq) edge [red, <->, bend left=10] node [right] {A} (q);
  \draw (q) edge [red, loop left] node {A} (q);
  \draw (p) edge [red, ->] node [above] {A} (pq);
  \draw (e) edge [red, ->] node [auto,sloped,xshift=2mm,yshift=-0.5mm] {A} (pq);
  \draw (p) edge [red, ->] node [auto,sloped,xshift=2mm,yshift=-0.5mm] {A} (q);
  \draw (pq) edge [blue, loop above] node {B} (pq);
  \draw (p) edge [blue, loop above] node {B} (p);
  \draw (q) edge [blue, ->, bend left=10] node [left] {B} (pq);
  \draw (e) edge [blue, ->] node [right] {B} (p);
\end{tikzpicture}
\hspace{1em}
\begin{tabular}{l|>{\centering\arraybackslash}m{42mm}>{\centering\arraybackslash}m{42mm}lll}
  \ & \textcolor{red}{Agent A} & \textcolor{blue}{Agent B} \\
  \midrule
  BDD & \aliceBDD & \bobBDD & \\[7em]
  MP & $(p\leftarrow \bot \cup p \leftarrow \top) ; (q \leftarrow \top)$ & $ p \leftarrow \top$ \\
\end{tabular}
\end{center}
\end{example}

Both approaches lead to a new model-checking problem, because when moving from one representation to another, the complexity of model checking can differ.
Interestingly, for the succinct representation theoretical results about the model checking complexity are available, but not yet for the symbolic structures.
With the present article we close this gap.
We summarise the known and new results in \Cref{tab:overview}. Note that when using symbolic structures, model checking \EL, \PAL and \DEL all become equally hard. We discuss why this happens after we introduce knowledge structures in \Cref{def-kns}.

\begin{table}[H]
  \centering
  \begin{tabular}{l|lll}
    Representation & Kripke models                        & Symbolic Structures            & Succinct models                 \\
    \midrule
    \EL-S5         & in P~\cite{rakhalpern1995}           & PSPACE: \autoref{thm:SymPA}    & -                               \\
    \PAL-S5        & in P~\cite{DELbook}                  & PSPACE: \autoref{thm:SymPA}    & -                               \\
    \DEL-S5        & PSPACE~\cite{dHvdP:CompDELS5}        & PSPACE: \autoref{thm:SymDELS5} & -                               \\
    \EL-K          & in P~\cite{rakhalpern1995}           & PSPACE: \autoref{thm:SymDEL}   & PSPACE~\cite{DLPABalbianiHST14} \\
    \PAL-K         & in P~\cite{ModalLogicBlackburnRV01}  & PSPACE: \autoref{thm:SymDEL}   & PSPACE~\cite{CharSucc2017}      \\
    \DEL-K         & PSPACE~\cite{ComplexityDELAucherS13} & PSPACE: \autoref{thm:SymDEL}   & PSPACE~\cite{CharSucc2017}      \\
  \end{tabular}
  \caption{Model checking complexity results. We write ``PSPACE'' for PSPACE-complete here.}\label{tab:overview}
\end{table}

While the \EL and \PAL rows in \autoref{tab:overview} show an easier complexity for Kripke models than for symbolic structures, the key difference lies in input size: symbolic structures can represent Kripke models exponentially more succinctly. Thus, polynomial-time algorithms for Kripke models may become exponential when measured against symbolic input.
On the other hand, the \DEL rows show that symbolic structures do not lead to an additional complexity blowup.
The results suggest that one should not translate Kripke models to other representations, but work directly with the symbolic or succinct representation, i.e.\ a given system description should be formalized directly into a structure with BDDs or mental programs, to avoid ever using exponential memory that would be needed by a Kripke model.

The symbolic structures have been implemented and benchmarks show that they outperform standard Kripke models, but the succinct representations have not been implemented, with the exception of~\cite{Hartlief2020}.
To better understand the relation between the different encodings of relations, we provide translations from BDDs to mental programs and vice versa, and study the translation complexity.

In the remainder of this section we define the languages.
In \autoref{sec:symbolic} we prove that symbolic model checking is PSPACE-complete, first for \PAL-S5, then for \DEL-K.
In \autoref{sec:succinct} we recall mental programs used in succinct DEL, 
and in \autoref{sec:comparison} we provide translations between BDDs and mental programs.
We conclude with comments about multi-pointed models and ongoing implementation work in \autoref{sec:conclusion}.

\begin{definition}
    Throughout this article we work with finite vocabularies $V$ that come with an order.
    For any vocabulary $V$ we define \emph{Boolean formulas} $\beta \in \LL_B(V)$ by
    $\beta ::= \top \mid \bot \mid p \mid \lnot\beta \mid \beta \land \beta$
    where $p \in V$.
    A \emph{state} is a subset of the vocabulary $s \subseteq V$.
    That is, we identify states with the atomic propositions that are true at them.
    The Boolean semantics are defined as usual and denoted by $s \vDash \beta$.
\end{definition}

To work with both \PAL and \DEL, we will use a general language definition with a parameter set $D$.

\begin{definition}
For any vocabulary $V$
and any set $D$
the language $\mathcal{L}_{D}(V)$ is given by
$\phi ::= \top \mid \bot \mid p \mid \lnot \phi \mid \phi \land \phi \mid K_i \phi \mid [d]\phi$
where $p \in V$, $i \in A$ and $d \in D$.
Concretely, $\LL_\EL$ is given by $D = \varnothing$,
$\LL_\PAL$ is given by $D = \{ !\phi \mid \phi \in \mathcal{L}_\PAL \}$,
and $\LL_\DEL$ is given by letting $D$ be the set of events from \autoref{d:trf}.
\end{definition}

The simpler logic we consider is Public Announcement Logic (\PAL), where the only actions are truthful public announcements, i.e.~$D$ is the set of all formulas, used as announcements.
More general actions include private announcements or factual (also called ontic) changes.
In \DEL these are usually modeled by letting $D$ be the set of all pointed action models~\cite{BMS1998:PAckPS,DELbook}, but here we will mostly use the transformers from~\cite{vBEGS2018:SymDEL}.
Yet another option are the succinct event models from~\cite{CharSucc2017}.
All these languages come with a mutual but well-founded recursion: elements of $D$ are made using formulas, and formulas use elements of $D$.
We refer to~\cite[Chapter~6]{DELbook} for details.

The following definition will be relevant to determine the input size of model checking problems.

\begin{definition}\label{def:sizeOfFormula}
  We define the length $|\cdot| \colon \LL_{\PAL}(V) \cup \LL_{\DEL}(V) \to \mathbb{N}$ of formulas as follows
  \[
    \begin{array}{lll}
      |\top| & := & 1 \\
      |\bot| & := & 1 \\
      |p| & := & 1 \\
    \end{array}
    \hspace{1cm}
    \begin{array}{lll}
      |\lnot\phi| & := & |\phi| + 1\\
      |\phi_1 \land \phi_2| & := & |\phi_1| + |\phi_2| + 1 \\
      |K_i \phi| & := & |\phi| + 1\\
    \end{array}
    \hspace{1cm}
    \begin{array}{lll}
      |[!\psi]\phi| & := & |\psi| + |\phi| + 1 \\
      |[\X,x]\phi| & := & |\X| + |\phi| + 1 \\
      \ \\
    \end{array}
  \]
  where $|\X|$ denotes the size of a transformer as in \autoref{defi:sizeOftrf} below.
\end{definition}

\section{Symbolic Structures with BDDs}\label{sec:symbolic}

We recall the definition and main features of the knowledge and belief structures from~\cite{vBEGS2018:SymDEL,GattingerThesis2018}.
In general we assume that elements of $\LL_B$ are represented using Binary Decision Diagrams (BDDs) from~\cite{Bryant86}.

\begin{definition}\label{def-kns}
  A \emph{knowledge structure} is a tuple $\F = (V,\theta,O)$ where
  $V$ is a finite set called \emph{vocabulary},
  $\theta \in \LL_B(V)$ is the \emph{state law} and
  $O_i \subseteq V$ for each $i \in A$ are the \emph{observables}.
  A \emph{state} of $\F$ is a $s \subseteq V$ such that $s \vDash \theta$.
  A \emph{pointed} knowledge structure is a tuple $(\F,s)$ where $s$ is a state of $\F$.
\end{definition}

Knowledge structures encode the knowledge of agents (that in Kripke models is represented by a relation) with a set of observables.
Each $O_i\subseteq V$ encodes a relation $R_i \subseteq \cP(V) \times \cP(V)$ given by $R_i s t \iff s\cap O_i = O_i\cap t$.
Note that this is always an equivalence, matching S5 Kripke models. Hence the corresponding frameworks are called \EL-S5, \PAL-S5 and \DEL-S5.
The more general belief structures (for \EL-K, \PAL-K and \DEL-K) can encode arbitrary relations using formulas over a double vocabulary. In addition to this, note that in~\Cref{tab:overview}, even though model checking for \PAL and \EL were much easier than \DEL (P for \PAL,\EL whereas PSPACE for \DEL) using Kripke models, when it comes to model checking using symbolic structures, all three of them becomes equally hard (PSPACE). This is because all possible states in a Kripke model are listed distinctly when the model is represented in a general Kripke structure. Therefore, going over all of them does not cost any computational resource beyond the size of the Kripke model itself. But when represented using symbolic structures, only the set of propositions $V$ is part of the input. Hence in order to iterate through indistinguishable possibilities of an agent (for example for checking formulas with $\hat{K}_i$ modalities), all possible (at most $2^{|V|}$) states that come out of $V$ have to be traversed and checked. Since such $\hat{K}_i$ modalities are present in \EL, \PAL as well as \DEL, model checking all three of them becomes equally hard.

\begin{definition}\label{def-bls}
  A \emph{belief structure} is a tuple $\F = (V,\theta,\Omega)$ where
  $V$ and $\theta$ are as in Definition~\ref{def-kns} and
  for each agent $i$ we have an \emph{observation law} $\Omega_i \in \LL_B(V \cup V')$
  where $V'$ denotes a fresh copy of $V$.
\end{definition}

The intuition behind $\Omega_i$ is that it is true at a pair of states iff the states are related.
That is, a Boolean formula $\Omega_i \in \LL_B(V \cup V')$ encodes the relation $R_i \subseteq \mathcal{P}(V) \times \mathcal{P}(V)$ given by $R_i s t :\iff s \cup t' \vDash \Omega_i$.
For example, if $\Omega_{\text{Alice}} = p \land q'$ then from any state where $p$ is true Alice will consider any other state where $q$ is true possible.
In particular we will have a loop at the state $\{p,q\}$ because $ \{p,q\} \cup \{p,q\}' \vDash p \land q'$.
This encoding is also widely used for temporal model checking~\cite[Section 8.3]{ChakiBDDbasedMC2018}.

Analogous to how symbolic structures encode Kripke models, transformers encode action models.
To keep track of old propositional values before factual change, we introduce another set of fresh variables, denoted by the operation ${(\cdot)}^\circ$.
Just as $V'$ is a fresh copy of $V$, $V_-^\circ$ is a fresh copy of $V_-$.

\begin{definition}\label{d:trf}
A \emph{transformer} is a tuple $\X = (V^+, \theta^+, {V_-}, {\theta_-}, {\Omega_i^+})$ where
$V^+$ is such that $V \cap V^+ = \varnothing$,
$\theta^+ \in \LL_\DEL(V \cup V^+)$ is called the \emph{event law},
$V_- \subseteq V$ is called the \emph{modified subset},
$\theta_- : V_- \rightarrow \mathcal{L}_{B}(V \cup V^+)$ is called the \emph{change law}
and $\Omega_i^+ \in \mathcal{L}_B(V \cup V')$ for each $i$.
An \emph{event} is a pair $(\X, x)$ where $x \subseteq V^+$.
To update $(\F, s)$ with $(\mathcal{\chi},x)$, let $\mathcal{F} \times \chi := (V^{new}, \theta^{new}, \Omega^{new}_i)$ where
$V^{new} = V \cup V^+ \cup V ^\circ _-$,
\[\theta^{new}=[V_- \mapsto V^\circ _-](\theta \land \|\theta ^+\|_\F) \land \bigwedge_{q \in V^-}(q \leftrightarrow [V_- \mapsto V^\circ _ -](\theta_-(q)))
\]
where $[p \mapsto \psi] \varphi$ denotes a substitution,
$\Omega^{new}_i=([V_- \mapsto V^{\circ}_-][(V_-)' \mapsto (V^{\circ}_-)']\Omega_i) \land \Omega_i^+$,
and lastly the new actual state is $s^x := (s \setminus V_-) \cup (s \cap V_-)^\circ \cup x \cup \{p \in V_- | s \cup x \vDash \theta_-(p)\}$.
\end{definition}

Since our goal is a complexity study of model-checking algorithms, we need to be precise about the size of all structures we take as inputs.
We define the size of transformers as follows.
\begin{definition}\label{defi:sizeOftrf}
    Given a transformer $\X = (V^+, \theta^+, {V_-}, {\theta_-}, {\Omega_i^+})$, we define its size
    by mutual recursion with \autoref{def:sizeOfFormula} as
    $|\X| :=
    |V^+|
    + |\theta^+|
    + |V_-|
    + \sum_{p \in V_-} |\theta_-(p)|
    + \sum_{i \in I} |\Omega_i^+|
    $.
\end{definition}

Also \autoref{d:trf} above and the following two definitions are mutually recursive.

\begin{definition}\label{d:symsemantics}
  The semantics for $\LL_\PAL$ and $\LL_\DEL$ on knowledge and belief structures are as follows.
  We omit the standard Boolean cases.
\begin{enumerate}
\item
  For knowledge structures:
  $(\F,s) \vDash {K}_i \phi$ iff
  for all states $t$ of $\F$, if $s\cap O_i=t\cap O_i $, then $(\F,t) \vDash \phi$.

  For belief structures:
  $(\F,s) \vDash {K}_i \phi$ iff
  for all states $t$ of $\F$, if $s \cup t' \vDash \Omega_i$, then $(\F,t) \vDash \phi$.
\item $(\F,s) \vDash [!\psi] \phi$ iff $(\F,s) \vDash \psi$ implies $(\F^\psi, s) \vDash \phi$
  where $\F^\psi := (V, \theta \land {\| \psi \|}_\F, O)$.
\item
  $(\F,s) \vDash [\X,x] \phi$ iff
  $(\F,s) \vDash [x \sqsubseteq V^+]\theta^+ \text{ implies } (\F \times \X , s^x) \vDash \phi$.
\end{enumerate}
\end{definition}

\begin{definition}\label{d:boolEquiv}
For any structure $\F$ and any formula $\phi \in \LL_{PAL}(V) \cup \LL_{DEL}(V)$
we define its \emph{local Boolean translation} $\| \phi \|_\F$ as follows.
\[
  {\| \top \|}_\F := \top
  \hspace{6mm}
  {\| \bot \|}_\F := \bot
  \hspace{6mm}
  {\| p \|}_\F := p
  \hspace{6mm}
  {\| \neg \psi \|}_\F := \neg \| \psi \|_\F
  \hspace{6mm}
  {\| \psi_1 \land \psi_2 \|}_\F := {\| \psi_1\|}_\F \land {\| \psi_2 \|}_\F
\]
If $\F$ is a knowledge structure, let
    ${\| K_i \psi \|}_\F := \forall(V \setminus O_i)(\theta \to {\| \psi \|}_\F)$.
If $\F$ is a belief structure, let 
    $\| K_i \psi \|_\F := \forall V' ( \theta' \to ( \Omega_i \to ({\| \psi \|}_\F)' ) )$.
Let
    ${\| [ \psi ] \xi \|}_\F := {\| \psi \|}_\F \to {\| \xi \|}_{\F^\psi}$,
    where $\F^\psi := (V, \theta \land {\| \psi \|}_\F, O)$.
Lastly, let
    ${\| [\X,x]\phi \|}_\F :=
      {\| [x \sqsubseteq V^+] \theta^+ \|}_\F
        \to
          [{V_-^\circ} \mapsto {V_-}]
            [x \sqsubseteq V^+]
              [{V_-} \mapsto {\theta_-(V_-)}]
                {\| \phi \|}_{\F \times \X}$
    where $F \times \X$ is from \autoref{d:trf}.
\end{definition}

The above Boolean translation is equivalent, i.e.\ $\F,s \vDash \phi$ iff $s \vDash \|\phi\|_\F$.
The SMCDEL implementation~\cite{SMCDEL} uses the translation, as it is faster than using~\autoref{d:symsemantics}.
For any Kripke model there exists an equivalent symbolic structure and vice versa, and similarly for action models and transformers~\cite{GattingerThesis2018}.

There are now six different model checking problems:
The three languages $\LL_\EL$, $\LL_\PAL$ and $\LL_\DEL$ can each be interpreted on knowledge structures (S5) and on belief structures (K).
We first consider the easiest case \PAL-S5 and then the most general case \DEL-K.
Just like general Kripke models are more general than those for S5, the $\Omega_i$ can express anything that $O_i$ can do, i.e. every knowledge structure can be seen as a belief structure, replacing $O_i$, with the formula $\Omega_i := \bigwedge_{p \in O_i} (p \leftrightarrow p')$.
This needs roughly twice as much memory, resulting in no significant increase of the model-checking input size.

\subsection{Symbolic Model Checking \PAL-S5}

To define the model checking problem for symbolic structures we need to say how exactly the input is given.
For $V$ and $O_i$ this is obvious, but for $\theta$ we stress that we assume a BDD and not a formula.

\begin{definition}
The \emph{symbolic model checking task for PAL for S5} is the following.
Given a knowledge structure $\mathcal{F} = (V, \theta, O_i)$ where $\theta$ is a Boolean function encoded as a BDD,
an actual state $s$ and a formula $\phi \in \LL_\PAL(V)$ which may contain dynamic modalities, decide whether 
$\mathcal{F},s \vDash \phi$.
The input size is $|V| + |\theta| + |O_i| + |\phi|$ where $|\theta|$ is the node count of the BDD and $|\phi|$ is the length of the formula.
\end{definition}

We will show that this problem is PSPACE complete.
The following theorem shows that the problem is already PSPACE-hard for EL (i.e. without announcements).
This means that the symbolic representation matters, because model checking EL on standard Kripke models is in P~\cite{rakhalpern1995}, but it becomes harder when moving from Kripke models to symbolic structures.

\begin{theorem}\label{thm:SymELhard}
Model checking $\LL_\EL$ on knowledge structures is PSPACE-hard.
\end{theorem}
\begin{proof}
We reduce the evaluation of a Quantified Boolean Formula (QBF) to model checking $\LL_\EL$ on a knowledge structure as follows. Before going further into QBFs, let us introduce a notation. For a set of propositions $P = \{p_1, p_2,\ldots, p_k\}$, the notation $\forall P\equiv\forall p_1\forall p_2\ldots\forall p_k$. Same is true for $\exists P$.

Take any QBF (wlog.~in prenex form) $\psi = \forall P_1 \exists P_2 \ldots \forall P_{n-1} \exists P_n \phi$ where $\phi$ is a Boolean formula (possibly in CNF) over $P_1 \cup \ldots \cup P_n$. Note that there are no free variables in $\psi$.
Let $V = \bigcup_i P_i$ and let $\theta = \top$.
Let the set of agents be $\{1, \ldots, n\}$ and let $O_i := V \setminus P_i$ for each $i$.
Let $\mathcal{F} = (V, \theta, O_i)$ and let $s = \varnothing$.
Let ${\widehat K}_i \phi := \lnot K_i \lnot \phi$.
We now have the following equivalences:
  \[
  \begin{array}{rcll}
\psi \text{ is QBF-true}
  & \iff & \vDash \forall P_1 \exists P_2 \ldots \forall P_{n-1}\exists P_{n} \phi & \text{QBF truth definition} \\
  & \iff & s \vDash \forall P_1 \exists P_2 \ldots \forall P_{n-1} \exists P_n \phi & \text{no free variables} \\
  & \iff & s \vDash \forall (V \setminus O_1) \exists (V \setminus O_2) \ldots \forall (V \setminus O_{n-1}) \exists (V \setminus O_n) \phi & \text{by definition of $O_i$} \\
  & \iff & \mathcal{F},s \vDash K_1 \hatK_2 \ldots K_{n-1} \hatK_n \phi & \text{by \autoref{d:symsemantics}} \\
  \end{array}
  \]
  The last formula is of the same length as the given QBF.
  Hence we have a polynomial reduction from QBF-truth to model checking knowledge structures.
\end{proof}
In short, it is hard to do model-checking in polynomial time on symbolic structures since we need to go over every possible valuation on $V$ while evaluating $K_i\psi$, unlike in Kripke models where these valuations are enumerated explicitly in the model and thus part of the input size.

Given that $\LL_\EL$ is a fragment of $\LL_\PAL$ and $\LL_\DEL$ we have the following corollary.
\begin{corollary}\label{cor:symknowlPSPACEhard}
Model checking $\LL_\PAL$ and $\LL_\DEL$ on knowledge structures is PSPACE-hard.
\end{corollary}

Having shown hardness, we now turn to membership.
For simplicity we focus on \PAL in this section.
We will define an algorithm and then show its correctness and its memory usage.
Note that to stay in PSPACE we cannot use the common PAL reduction axioms (see \cite[Def 4.53]{DELbook}) as they would blow up the size of formulas by making copies of the announced formulas.
For example, consider $[! K_i p ] K_i p$.
The reduction axiom would give us the following equivalences:
\[
  [! [! K_i p ] K_i p] K_i p
  \equiv ([! K_i p ] K_i p)\to K_i[! [! K_i p ] K_i p]p
  \equiv ((K_i p)\to K_i[!K_i p]p)\to K_i[!((K_i p)\to K_i[!K_i p]p)]p
\]
Here the last formula which would then be model checked is huge in size with respect to the initial $[! [! K_i p ] K_i p] K_i p$.
In general, rewriting $[! [! [! \ldots [! [! K_i p] ] K_i p] K_i p] \ldots ] K_i p$
leads to a formula with $2^n$ copies of $K_i p$ and thus needing exponential space in terms of the original formula length.

Our Algorithm~\ref{alg:checkPAL} $\chk$ takes as inputs a knowledge structure $\F$, a $\PAL$ formula $\phi$ and a list of formulas $L$.
For the model checking problem $L$ is initially empty and populated by the algorithm itself.

\begin{algorithm}[H]
\caption{$\chk$}\label{alg:checkPAL}
\KwIn{knowledge structure $\F = (V, \theta, O)$,
    list of formulas $L = [\ell_0,\ldots,\ell_k]$,
    state $s \subseteq V$ of $\F \times \ell_0 \times \dots \times \ell_k$,
    formula $\phi$}
\Precondition{$s$ is a state of $\F$ (i.e.\ $s \vDash \theta$) and $L$ can be announced on $(\F,s)$.}
\KwOut{$\mathsf{true}$ or $\mathsf{false}$}

\Switch{$\phi$}{
    \lCase{$p$}{
        \Return $p \in s$
    }
    \lCase{$\lnot \phi$}{
        \Return NOT $\chk(\F, L, s, \phi)$
    }
    \lCase{$\phi_1 \land \phi_2$}{
        \Return ($\chk(\F, L, s, \phi_1)$ AND $\chk(\F, L, s, \phi_2)$)
    }
    \Case{$K_i \phi_1$}{
        \ForEach{$t \subseteq V$\label{algoline:PALmcAllValIter}
        }{            
            \If{$t \vDash \theta$ AND $t\cap O_i = s\cap O_i$\label{algoline:checkindisting}}{
                    \texttt{stillExists} := $\mathsf{true}$ \tcp{did $t$ ``survive'' the announcements in $L$?}
                    \ForEach{$j \in [0, \dots, k]$\label{algoline:verifytsurv}
                    }{
                        \lIf{NOT $\chk(\F, [\ell_0,\ldots,\ell_{j-1}], t, \ell_j)$ } {
                            \texttt{stillExists} := $\mathsf{false}$
                        }
                    }
                    \If{\texttt{stillExists} AND NOT $\chk(\F, [\ell_0,\ldots \ell_k], t, \phi_1)$}{
                            \Return $\Fa$
                    }
            }
        }
        \Return $\mathsf{true}$
    }
    \Case{$[!\psi]\phi_1$}{
        \lIf{$\chk(\F, L, s,\psi)$\label{algoline:pubannounceantecedent}}{
            \Return $\chk(\F, L ++ [\psi], s, \phi_1)$
            \textbf{else} \Return $\Tr$\label{algoline:pubannouncevacutrue}
        }
        
    }
}
\end{algorithm}

The algorithm proceeds by recursion on the formula $\phi$.
There are two challenges we need to take care of so that the algorithm takes only polynomial space:
\begin{itemize}
\item Firstly, we cannot compute the BDDs of announcements as in \autoref{d:boolEquiv}, because short formulas may have BDDs of exponential size.
  Hence we track announcements in the list $L$.
\item When evaluating formulas of the form $K_i\phi$, we must not compute the exponentially large set of all states.
  But we can iterate over them instead, without ever storing the full set.
\end{itemize}
To make both solutions compatible with each other, when evaluating $K_i\phi$ we check in Line~\ref{algoline:verifytsurv} whether a state still exists after the sequence of announcements in $L$, using recursive calls.
This technique is similar to the model checking algorithm for succinct DEL in~\cite{ComplexityDELAucherS13}.
It is also similar to the context-dependent semantics in~\cite{wang2013axiomatizations}, where the announcements are stored as a list of formulas and the model update is not triggered until $K_i \phi$ formulas are evaluated.
We now formally state that $\chk$ is correct.

\begin{restatable}{lemma}{PALcheckCorrect}\label{l:PALcheckCorrect}
  Given a pointed knowledge structure $\F$, a state $s$ of $\F$,
  a list of \PAL-S5 formulas $L = [\ell_0,\ldots,\ell_k]$ and a \PAL-S5 formula $\phi$,
  we have $\F^{\ell_0\ldots\ell_k},s\vDash\phi$ iff $\chk(\F,L,s,\phi)$ returns $\Tr$.
\end{restatable}
\noindent\emph{Proof sketch.}
By induction on the size of inputs $\phi$, $\F$ and $L$. The difficult case is for $K_i\phi_1$. It iterates over all valuations (line~\ref{algoline:PALmcAllValIter}) checking whether they agree on what $i$ observes (line \ref{algoline:checkindisting}) and recursively verifies $\phi_1$ as per \autoref{d:symsemantics}. The survival of the valuation by the sequence of announcements $L$ is checked in line~\ref{algoline:verifytsurv}.
For a detailed proof, see the appendix.

Next comes our main result for this subsection.

\begin{theorem}\label{thm:SymPA}
  Model checking $\LL_\PAL, \LL_\EL$ on knowledge structures is PSPACE-complete.
\end{theorem}
\begin{proof}
  Hardness follows from \autoref{cor:symknowlPSPACEhard}.
  By \autoref{l:PALcheckCorrect} it remains to show that \autoref{alg:checkPAL} only takes space polynomial in the size of the input. Membership for $\LL_\EL$ follows as a special case of $\LL_\PAL$.
  At any instance of a call to $\chk(\F,L,s,\phi)$, exactly one switch case matches.
  Every recursive call taken in $\chk(\F,L,s,\phi)$ is of size no more than the size of the current input.
  Moreover, in the call stack, there are at most linear many recursive calls left to evaluate at any given instance.
  In each call, the space used is polynomial.
  The most intensive is the case for $K_i \phi_1$.
  We iterate over every valuation $t \subseteq V$, which needs memory in the size of $|V$|.
  Note that we never store the whole list of states nor results.
\end{proof}

We note that the model checking algorithm performing well \emph{in practice} (and used in SMCDEL~\cite{SMCDEL}) is \emph{not} in PSPACE.
To see this, take a Boolean formula for which the BDD has to be large, for example consider the vocabulary $V = \{p_1, \dots p_{2n}\}$ and the formula
\label{largeBddForm}
$\phi := (p_1 \land p_{n+1} ) \lor \ldots \lor (p_n \land p_{n+n})$
which has length in $\mathcal{O}(n)$.
Suppose then we model check $[!\phi]p$ on a knowledge structure $\F$ with state law $\theta = \top$ using the SMCDEL algorithm.
This means we update the state law (stored as a BDD) from $\top$ to $\theta \land \|\phi\|_\F$.
As discussed in~\cite[p.~681]{Bryant86} the BDD of this formula needs $2^{n+1}$ many nodes.
Hence the resulting knowledge structure needs exponential space and this is not possible in PSPACE.

In summary, there is a trade-off between time and space here: \Cref{alg:checkPAL} will need exponential time to evaluate $K_i$, whereas the BDD-based method may need exponential space for $[!\phi]$.

\subsection{Symbolic Model Checking \DEL-K}

We now generalize the setting from the previous section in two aspects.
First, we switch from S5 to K by using belief structures instead of knowledge structures.
Second, we move from public announcements to general events.
We already discussed that \DEL can be seen as an extension of \PAL in the introduction.
We now illustrate how the move from S5 to K, and the additional feature of factual change are dealt with symbolically.
The difference between S5 and K is captured by $O_i \subseteq V$ and $\Omega_i \in \LL_B(V \cup V^+)$: $O_i$ always encodes an equivalence relation for a hard notion of knowledge, while $\Omega_i$ can encode arbitrary relations to describe belief.
Factual change is captured by $V_-$ and $\theta_-$ in the transformers from~\autoref{d:trf}, describing which propositions are changed and how~\cite[Section 2.8]{GattingerThesis2018}.

\begin{theorem}\label{thm:SymELKhard}
Model checking $\LL_{\EL}$ on belief structures is PSPACE-hard.
\end{theorem}
\begin{proof}
Similar to the proof of \autoref{thm:SymELhard}, but instead of $O_i$ let $\Omega_i$ be the BDD of $\bigwedge_{p \in P_i} (p \leftrightarrow p')$.
This BDD has size linear in $|V|$.
Hence we reduce QBF truth to model checking $\LL_\EL$ on belief structures.
\end{proof}

Again, given that $\LL_\EL$ is a fragment of $\LL_\PAL$ and $\LL_\DEL$ we have the following corollary.
\begin{corollary}\label{cor:symKPSPACEhard}
Model checking $\LL_\PAL$ and $\LL_\DEL$ on belief structures is PSPACE-hard.
\end{corollary}

\begin{algorithm}[H]
\caption{$\chkfcdelK$}\label{alg:checkfcDELK}
\Input{Knowledge structure $\F=(V,\theta, \Omega)$,
    list of events $L = [(\X_0, x_0), \ldots, (\X_k, x_k)]$,
    state $s \subseteq V$,
    formula $\phi \in \LL_{DEL}(V)$}
\Precondition{$s$ is a state of $\F$ (i.e.\ $s \vDash \theta$) and $L$ can be executed on $(\F,s)$.}
\KwOut{$\mathsf{true}$ or $\mathsf{false}$}

\Switch{$\phi$}{
    \Case{$p$}{
        \If{$L = []$}{
            \Return $p \in s$
        }
        \tcp{Compute the final state}
        \ForEach{$j \in [0, \dots, k]$}{
            $s := (s \setminus V_{j,-}) \cup \{ q \in V_{j,-} \mid \chkfcdelK(\F, [(\X, x_0), \ldots, (\X,x_j)], s, [x_j^+ \sqsubseteq V_j ^+] \theta_j^-(q)) \}$
        }
        \Return $p \in s$
    }
    \Case{$\lnot \phi$}{
        \Return NOT $\chkfcdelK(\F, L, s, \phi)$
    }
    \Case{$\phi_1 \land \phi_2$}{
        \Return ($\chkfcdelK(\F, L, s, \phi_1)$ AND $\chkfcdelK(\F, L, s, \phi_2)$)
    }
    \Case{$K_i \psi$}{
        \ForEach{$t \subseteq V$, $t^+_0 \subseteq V_0^+, \dots, t^+_k \subseteq V_k^+$}{
            \If{$t \vDash \theta$}{
                \If{$s \cup t^\prime \vDash \Omega_i $ AND
                $x_1^+ \cup t^{+\prime}_1 \vDash \Omega_{0,i}^+$ AND
                \dots, $x_k^+ \cup t_k^{+\prime} \vDash \Omega_{k,i}^+$
                \label{algoline:DELcheckindisting}}{
                    \texttt{accessible} := $\mathsf{true}$ \tcp{did $t$ ``survive'' the actions in $L$?}
                    \ForEach{$j \in [0, \dots, k]$}{
                        \If{not $\chkfcdelK(\F, [(\mathcal{X}_0, t_0^+), \ldots, (\mathcal{X}_{j-1},t_{j-1}^+)], t, [t_j^+ \sqsubseteq V_j ^+] \theta_j ^+)$}{
                            \texttt{accessible} := $\mathsf{false}$
                        }
                    }
                    \If{\texttt{accessible}}{
                        \If{not $\chkfcdelK(\F, [(\mathcal{X}_0, t^+_0), \ldots, (\mathcal{X}_k,t^+_k)], t, \psi)$}{
                            \Return $\Fa$
                        }
                    }
                }
            }
        }
        \Return $\mathsf{true}$
    }
    \Case{$[\mathcal{X}, x] \psi$
      }{
        \If{$\chkfcdelK(\F, L, s, [x \sqsubseteq V ^+] \theta ^+)$ \tcp{checking the precondition}\label{algoline:DELeventprecondition}}{
            \Return ($\chkfcdelK(\F, L ++ [( \mathcal{X}, x )], s, \psi)$)
        }
        \Return $\Tr$\label{algoline:DELvacutrue}
    }
}
\end{algorithm}

Similar to the previous section about \PAL-S5, we prove membership by defining an algorithm.
Also here the challenge is that we cannot compute the potentially too large BDD of the new state law after updating.
Fortunately, also a similar solution works: the list $L$ will now not just contain a list of public announcements, but a list of transformers.
Then to check $\F,s \vDash [\mathcal{X},x]\phi$ we add that transformer to $L$.
Later, when we check $K_i \phi$, we iterate ``on the fly'' over the states $t$ that may be the result of the sequence of transformers, including the additional atoms introduced and taking into account the factual change. 

\begin{restatable}{lemma}{DELcheckCorrect}\label{l:DELcheckCorrect}
  Given a pointed knowledge structure $\F = \ldiaarg{V,\theta,O_i},s\subseteq V$ and a $\DEL$ formula $\phi$, we have
  $\F,s \vDash \phi$ iff $\chkfcdelK(\F,[],s,\phi)$ returns $\Tr$.
\end{restatable}
\noindent\emph{Proof sketch.}
Similar to \autoref{l:PALcheckCorrect}, see appendix for the full proof.

\begin{theorem}\label{thm:chkDELKpolySpace}
    $\chkfcdelK(\F,[],s,\phi)$ takes at most polynomial space with respect to size of input.
\end{theorem}
\begin{proof}
    An argument analogous to that for \autoref{thm:SymPA} shows that \autoref{alg:checkfcDELK} runs using space at most polynomial in size of the input.
\end{proof}

Together we now get our main result for \DEL-K and as a result \PAL-K and \EL-K as well.

\begin{theorem}\label{thm:SymDEL}
  Model checking $\LL_\EL, \LL_\PAL$ and $\LL_\DEL$ on belief structures is PSPACE-complete.
\end{theorem}
\begin{proof}
  Using \autoref{alg:checkfcDELK}, which is correct by \autoref{l:DELcheckCorrect} it is correct and uses polynomial space by \autoref{thm:chkDELKpolySpace}, we have membership. By \autoref{cor:symKPSPACEhard}, we have hardness.
\end{proof}
 
 Moreover, since \DEL-S5 is a special case of \DEL-K we have the following.
 
 \begin{theorem}\label{thm:SymDELS5}
     Model checking $\LL_{\text{\DEL}}$ on knowledge structures is PSPACE-complete.
 \end{theorem}
 \noindent\emph{Proof sketch.}
 Just like in \autoref{thm:SymELKhard}, for membership use the equivalent $\Omega_i = \bigwedge_{p\in O_i}p\leftrightarrow p^\prime$ (for the belief structure) and $\Omega^+_i = \bigwedge_{p^+\in O^+_i}p^+\leftrightarrow p^{+\prime}$ (for any transformer).
 The resulting belief structure is at most twice the size of the given knowledge structures.
 The hardness result is from \autoref{thm:SymELhard}.

\section{Succinct Models with Mental Programs}\label{sec:succinct}

An alternative to the symbolic structures used in the previous section is the framework of \emph{Succinct} \DEL as presented by~\cite{CharSucc2017}.
It is based on a version of (Propositional Dynamic Logic) PDL, also called Dynamic Logic of Propositional Assignment (DLPA)~\cite{BalbianiDynamic2013}.

\begin{definition}
The language of \emph{mental programs} over a vocabulary $V$ is defined by
\[
  \pi ::= p\leftarrow\top \mid p\leftarrow\bot \mid \beta? \ \mid \pi\union\pi \mid \pi;\pi \mid \pi\cap\pi 
\]
where $p\in V$ and $\beta$ is a Boolean formula over $V$.
Let $\Pi_V$ denote the set of all mental programs over $V$.
The length of a mental program is defined as follows.
\[
\begin{array}{lll}
 |p \leftarrow \top| & := & 1 \\
 |p \leftarrow \bot| & := & 1 \\
\end{array}
\hspace{1cm}
\begin{array}{lll}
 |\beta?|            & := & |\beta| \\
 |\pi_1 \cup \pi_2|  & := & |\pi_1| + |\pi_2| \\
\end{array}
\hspace{1cm}
\begin{array}{lll}
 |\pi_1 ; \pi_2|     & := & |\pi_1| + |\pi_2| \\
 |\pi_1 \cap \pi_2|  & := & |\pi_1| + |\pi_2| \\
\end{array}
\]
\end{definition}

In the rest of the article we are mostly interested in how a single relation over the set of states $\powerset{V}$ can be encoded.
Hence we omit further details about succinct DEL and refer to \cite{CharSucc2017} for how Kripke models can be encoded using one mental program for each agent, and how actions and events can be encoded using mental programs as well.

\begin{definition}\label{def:mpSemantics}
  Two states $s, t \subseteq V$ are related by a mental program $\pi$ (written as $s\xrightarrow{\pi}t$) as follows:
  \[
    \begin{array}{ll}
      s\xrightarrow{p\leftarrow\top} t & :\iff t = s\union\{p\}\\
      s\xrightarrow{p\leftarrow\bot} t & :\iff t = s\setminus\{p\}\\
      s\xrightarrow{\beta ?} t & :\iff s=t\mbox{ and }s\vDash\beta\\
    \end{array}
    \hspace{1cm}
    \begin{array}{lll}
      s\xrightarrow{\pi_1\union\pi_2}t & :\iff s\xrightarrow{\pi_1}t\mbox{ or }s\xrightarrow{\pi_2} t\\
      s\xrightarrow{\pi_1;\pi_2} t & :\iff \exists u \subseteq V :~ s\xrightarrow{\pi_1}u\mbox{ and }u\xrightarrow{\pi_2}t\\
      s\xrightarrow{\pi_1\cap\pi_2} t & :\iff s\xrightarrow{\pi_1}t\mbox{ and } s\xrightarrow{\pi_2} t
    \end{array}
  \]
  We define $R_\pi := \{ (s,t) \in V \times V \mid s\xrightarrow{\pi}t \}$.
\end{definition}

\begin{example}
From the state $\{p,q\}$ we can reach the state $\{p\}$ using the mental program $q \leftarrow \bot$.
From the state $\{p,q\}$ we can reach the states $\{p\}$ and $\{q\}$ using the mental program $(p \leftarrow \bot) \cup (q \leftarrow \bot)$.
\end{example}

Other versions of mental programs~\cite{CharSucc2017} also include the general assignment $p \leftarrow \beta$ and inverse $\pi^{-1}$.
However, these operators do not add expressivity, as shown in~\cite{GatSymSucShifts2020}.
Also the following shows that mental programs as defined above are complete in the sense that they can encode all relations.

\begin{definition}\label{d:ofgoto}
    For any $x \subseteq y \subseteq V$, let
    $\outof(x,y) := \bigwedge_{p \in x} p \land \bigwedge_{p \in y \setminus x} \lnot p$.
    For any set $x \subseteq V$ we define
    $\change(x) := \bigsem_{p \in x}((p \leftarrow \top) \cup (p \leftarrow \bot))$,
    and for any state $s \subseteq V$ we define
    $\goto(s,V) := (\bigsem_{p \in s}(p \leftarrow \top)) \ \ ; (\bigsem_{p \in V \setminus x}(p \leftarrow \bot))$.
\end{definition}

We illustrate \Cref{d:ofgoto} with some examples.
First, $x=\{p\}$ and $y=\{p,q,r\}$ then we have $\outof(x,y) = p \land \lnot q \land \lnot r$.
Second, if $x = \{q,r\}$ then $\change(x) = (q \leftarrow \top \cup q \leftarrow \bot); (r \leftarrow \top \cup r \leftarrow \bot)$.
Third, if $V = \{p,q,r\}$ then $\goto(\{p\},V) = p \leftarrow \top; q \leftarrow \bot; r \leftarrow \bot$.

\begin{lemma}\label{lemma:mpCanExpressAll}
  For any relation $R \subseteq V \times V$ there is a mental program $\pi$ such that $R_\pi = R$.
\end{lemma}
\begin{proof}
    Let $\pi := \bigcup_{(x,y) \in R} ( \outof(x, V)? ; \goto(y) )$ and apply \autoref{def:mpSemantics}.
\end{proof}

In~\cite{CharSucc2017} it is shown how Kripke models can be encoded with mental programs and that model checking \DEL on such \emph{succinct} models is in PSPACE.
This is shown not by giving a PSPACE algorithm, but an alternating Turing machine algorithm and then using the fact that PSPACE = APTIME~\cite{CKS1981:Alternation}.
This means the algorithm given cannot easily be translated to an actual implementation.
The authors suggest that ``a model checking procedure for our succinct language may use BDD techniques of \cite{vBEGS215:SymDEL-LORI}''~\cite[p. 130]{CharSucc2017}.

\section{Comparison and Translations}\label{sec:comparison}

To compare BDDs and mental programs, we present a list of examples in \Cref{tab:examples}.
Each row stands for a class of relations or an operation on relations, and shows how it can be represented or executed, in the different representations.
Notably, observing more propositions, and thus going from the total relation (in row 2) to the identity (in row 5), leads to a shorter mental program, but to a longer Boolean formula.

\begin{table}[H]
  \centering
  \begin{tabular}{llll}
     & Mental program $\pi$                              & Boolean function $\Omega$                           & Observable $O$ \\
    \midrule
    empty relation          & $?\bot$                                           & $\bot$                                              & n/a            \\
    total relation          & $\change(V)$                                      & $\top$                                              & $\varnothing$  \\
    only observe $p$        & $\change(V \setminus \{p\})$                      & $p \leftrightarrow p'$                              & $\{p\}$        \\
    only observe $\{p, q\}$ & $\change(V \setminus \{p,q\})$                    & $(p\leftrightarrow p')\wedge (q\leftrightarrow q')$ & $\{p, q\}$     \\
    identity relation       & $?\top$                                           & $\bigwedge_{p \in V} (p \leftrightarrow p')$        & $V$            \\
    single edge $s \to t$   & $? \outof(s,V); \change(V); \outof(t,V)$          & $\outof(s, V) \land \outof(t', V')$                 & n/a            \\[0.5em]
    complement              & n/a                                               & Given $\beta$, use $\lnot \beta$                    & n/a            \\
    inverse                 & See translation in~\cite{GatSymSucShifts2020}.    & Given $\beta$, swap $p$ with $p'$ etc.              & no change      \\
    composition             & Given $\pi_1$ and $\pi_2$, use $\pi_1;\pi_2$      & See Def.~\ref{def:tr}                               & n/a            \\
    intersection            & Given $\pi_1$ and $\pi_2$, use $\pi_1 \cap \pi_2$ & Given $\beta_1$ and $\beta_2$, use $\beta_1 \land \beta_2$ & union   \\
  \end{tabular}
  \caption{Examples of relations and operations.}\label{tab:examples}
\end{table}

The ``n/a'' entries in \Cref{tab:examples} mean that this operation cannot be defined (in general) with the given representation.
For example, given a mental program we cannot easily define one for the complement (Row 7).
All mentioned operations can be done with Boolean functions, but with observational variables (last column) three operations are impossible because they yield non-equivalence relations.

\subsection{Translating Mental Programs to BDDs}

Given a mental program we now want to encode the same relation over subsets of $V$ using a Boolean function.
We give the following definition using Boolean formulas on the right side, but it also provides a mapping of mental programs to BDDs by reading all connectives on the right side as BDD operations.
Recall that $\Pi_V$ is the set of mental programs over $V$ and $\LL_B(V \cup V')$ is the set of Boolean formulas over the double vocabulary where the prime makes a fresh copy of each atomic proposition letter.

\begin{definition}\label{def:tr}
We now define a function $\tr: \Pi_V \rightarrow \LL_B(V \cup V')$.
\[
    \begin{array}{ll}
    \tr(p\leftarrow\top)  &\!\!\!\!\!:= p' \land \bigwedge_{p \neq q \in V} (q \leftrightarrow q') \\
    \tr(p\leftarrow\bot)  &\!\!\!\!\!:= \neg p' \land \bigwedge_{p \neq q \in V} (q \leftrightarrow q') \\
    \tr(\beta ?)          &\!\!\!\!\!:= \beta \land \bigwedge_{p \in V} (p \leftrightarrow p')
    \end{array}
    \begin{array}{ll}
    \tr(\pi_1\union\pi_2) &\!\!\!\!\!\!:= \tr(\pi_1) \vee \tr(\pi_2)\\
    \tr(\pi_1;\pi_2) &\!\!\!\!\!\!:= 
      [V'' \!\!\mapsto\!\! V'] (
        \exists V' (
          \tr(\pi_1)
          \land
          \tr(\pi_2)'
        ) 
      )\\
    \tr(\pi_1\cap\pi_2) &\!\!\!\!\!\!:= \tr(\pi_1)\wedge\tr(\pi_2)
    \end{array}
\]
\end{definition}

The left cases of \autoref{def:tr} are easy, they mostly consist of the ``do not change anything else'' conjunctions.
Among the right three cases, the one standing out is the composition $\pi_1;\pi_2$.
Here we write $[\cdot \mapsto \cdot]\phi$ for simultaneous substitution of atoms in $\phi$.
For example, $[\{p',q'\}\mapsto\{p,q\}](p' \land q'') = (p \land q'')$.
Strictly speaking in $[A \mapsto B]\phi$ both $A$ and $B$ are ordered lists and we use the implicit bijection between them as the substitution function~\cite[\href{https://malv.in/phdthesis/gattinger-thesis.pdf\#theorem.1.0.3}{Def.~1.0.3}]{GattingerThesis2018}.
We note that $\tr(\pi_2)'$ is $[V\!\!\mapsto\!\! V'][V'\!\!\mapsto\!\! V'']\tr(\pi_2)$, i.e.\ this changes the formula from the vocabulary $V \cup V'$ to the vocabulary $V' \cup V''$.
The Boolean quantification $\exists V'$ then eliminates all single-primed variables and lastly the outermost substitution ensures that the resulting formula is over the vocabulary $V \cup V'$ as we want.
We stress that $\exists$ is only expensive when done syntactically.
As a BDD operation it in fact deletes variables and nodes from a BDD~\cite{Bryant86}.

\begin{restatable}{theorem}{Correcttranslationbddtomp}\label{thm:Correcttranslationbddtomp}
The translation from Def.~\ref{def:tr} is correct: for any $s,t\subseteq V$, we have $s \xrightarrow{\pi} t$ iff $(s \cup t') \vDash \tr(\pi)$.
\end{restatable}
\noindent\emph{Proof sketch.}
By induction on the structure of $\pi$ and applying \autoref{def:mpSemantics}. See appendix for proof.

Note that already translating $p \leftarrow \top$ needs $\mathcal{O}(|V|)$ many nodes in the BDD.
This illustrates the high length of formulas resulting from this translation.
However, it is not a problem specific to the particular translation given above, but applies to any correct translation, as the following theorem states.

\begin{restatable}{theorem}{Polyreductionmptobdd}\label{thm:Polyreductionmptobdd}
For any translation $\tr'$ from mental programs to BDDs there exists a mental program $\pi$ such that $\tr'(\pi)$ has size exponential in $|\pi|$.
\end{restatable}

\noindent\emph{Proof sketch.}
  Take any vocabulary $V = \{p_1, \dots p_{2n}\}$ with this ordering fixed. Let $\pi := \beta ?$ where $\beta$ is the Boolean formula $(p_1 \land p_{n+1} ) \lor \ldots \lor (p_n \land p_{n+n})$ from page~\pageref{largeBddForm}.
  We show that the BDD corresponding to the mental program $\pi$ has at least $2^{n+1}$ many nodes.
  For the details, see the appendix.

\subsection{Translating BDDs to Mental Programs}

We write $(t_0 \dashleftarrow (p_i) \rightarrow t_1)$ to denote a node in a BDD that is labelled with variable $p$ and has an else-edge (dotted) pointing to node $t_0$ and a then-edge (solid) pointing to node $t_1$.
Note that $p_i$ may come from $V$ or from $V'$, and in the latter case we write $p_i'$.
For the leaves, we just write $\top$ and $\bot$.

\begin{definition}\label{def:btom}
We define a function
$\tau: BDD_{(V \cup V')} \times [V \cup V'] \to \Pi_V$ where $[V \cup V']$ is the set of all lists with elements from $V \cup V'$.
We distinguish different cases for the given BDD $\Omega$.
\[
    \tau(\bot, L)  := ? \bot
    \hspace{0.8cm}
    \tau(\top, [])  := ? \top
    \hspace{0.8cm}
    \tau(\top, [p_k, .., p_n]) := ((p_k \leftarrow \bot) \union (p_k \leftarrow \top)) ; \tau(\top,[p_{k+1}, .., p_n])
\]
\vspace*{-2em}
\begin{alignat*}{2}
    \tau(\Omega = (t_0 \dashleftarrow (p_i) \rightarrow t_1),[p_k, .., p_n]) &:=&&
      \twopartdef
      {(? \neg p_k;\tau(t_0,[p_k, .., p_n])) \union (? p_k;\tau(t_1,[p_k, .., p_n]))}{\text{if } i=k}
      {((p_k \leftarrow \bot) \union (p_k \leftarrow \top)) ; \tau(\Omega,[p_{k+1}, .., p_n])}{\text{otherwise }}
    \\
    \tau(\Omega = (t_0 \dashleftarrow (p_i') \rightarrow t_1), [p_k, .., p_n]) &:=&& \twopartdef
                                                                                     {\begin{array}{@{}l}
                                                                                       (p_k \leftarrow \bot;\tau(t_0,[p_{k+1}, .., p_n]))
                                                                                       \\
                                                                                       \union (p_k \leftarrow \top;\tau(t_1,[p_{k+1}, .., p_n]))
                                                                                     \end{array}
                                                                                     }{\text{if } i = k}
        {((p_k \leftarrow \bot) \union (p_k \leftarrow \top)) ; \tau(\Omega,[p_{k+1}, .., p_n])}{\text{otherwise}}
\end{alignat*}
Lastly, define $\tau_0 : BDD_{V \cup V'} \to \Pi_V$ by $\tau_0(\Omega) := \tau(\Omega, [p_0,\ldots,p_n])$.
\end{definition}

Both ``otherwise'' cases above are about variables $i > k$ not mentioned by BDD.
If the BDD does not mention a unprimed variable, this simply means the relation does not depend on that variable in the starting state, so the mental program also does not have to mention it.
On the other hand, not mentioning a primed variable in the BDD means we must allow the mental program to change that variable arbitrarily.

\begin{theorem}
The translation $\tau$ always terminates.
\end{theorem}
\begin{proof}
We observe that at each recursive step either the size of the BDD $\Omega$ or the length of the list $L$ strictly decreases, while the other size stays the same.
\end{proof}

\begin{restatable}{theorem}{Bddtompcorrect}\label{thm:Bddtompcorrect}
  The translation $\tau_o$ from \Cref{def:btom} is correct.
  That is, given a vocabulary $V$, for any $s,t \subseteq V$, we have $s \xrightarrow{\tau_0(\Omega)} t$ iff $(s \cup t') \vDash \Omega$.
\end{restatable}
\noindent\emph{Proof sketch.}
For any BDD $\Omega$, let $\Omega^\ast$ be its unraveling to a tree.
Note that we have $\tau(\Omega) = \tau(\Omega^\ast)$.
Hence for the proof we assume wlog.~that $\beta$ is a tree and proceed by induction over the tree structure.
For all details, see the appendix.

We now consider the length of the mental programs resulting from the translation.
Unfortunately, given a BDD for a vocabulary $V$ with $|V| = n$ the output of $\tau$ will always have at least length $\mathcal{O}(2^n)$.
The result may be simplified using equivalences such as $p \leftarrow \top; p?  \equiv  p \leftarrow \top$ and $p \leftarrow \bot; p?  \equiv  \bot?$.
But even including such simplifications we believe that there is no better translation in the following sense.

\begin{conjecture}\label{conj:mpToBddLarge}
For any translation $\tau'$ from BDDs to mental programs that is correct in the sense of \autoref{thm:Bddtompcorrect}, there exists a BDD $\beta$ such that $\tau'(\beta)$ has a length exponential in the number of nodes of $\beta$.
\end{conjecture}

As an example, take the BDD encoding the relation
$R = \{ (s,s) \mid |s| = \lceil\frac{|V|}{2}\rceil \}$,
i.e.~the identity restricted to the states where exactly half of the vocabulary is true.
The BDD encoding $R$ is a $|V|\cdot|V|$ grid.
We can find a mental program encoding this relation using (the proof of) \autoref{lemma:mpCanExpressAll},
but the size of the mental program is $2^{|V|/2} \cdot |V|$.
It is not clear whether there exists a shorter mental program that encodes $R$.
Essentially the mental program needs to count, and at the same time preserve the valuation.
But as there are no additional variables the only way to count without forgetting is to try all valuations.
Thus we believe that any mental program $\pi$ encoding $R$ must have size exponential in the size of the $BDD$.

\section{Conclusions and Future Work}\label{sec:conclusion}

In this paper, we showed that the symbolic model-checking tasks for \EL, \PAL and \DEL on both knowledge structures and belief structures are all PSPACE-complete.
In addition we compared how the same relations can be encoded in belief structures with BDDs on one hand, and in succinct models with mental programs (a regular language like syntax) on the other hand.
That is, we provided translations from mental programs to BDDs, and back.
We have also proven that any such translation from mental programs to BDDs will lead to an exponential blowup in size, and we conjecture the same for the other direction from BDDs to mental programs.

An aspect we did not consider here but that is relevant for epistemic planning are multi-pointed models and actions.
Concerning implementations, a \PAL model checker using mental programs was implemented in~\cite{Hartlief2020}.
In parallel to the theoretical work presented here we have improved the performance of this code and implemented the translations from \autoref{sec:comparison}.
We plan to benchmark all methods and eventually merge them into SMCDEL in the future.

\bibliographystyle{eptcs}
\bibliography{refs}

\clearpage
\appendix

\section{Proofs}
\PALcheckCorrect*
\begin{proof}
  By induction on the input size and case distinction on $\phi$.
  We only show the interesting cases for knowledge and for announcements.

  \smallskip
  \noindent\textbf{Case} $\phi = K_i \phi_1$.
  We have:
  \begin{align*}
    \F^{\ell_0\ldots\ell_k},s\vDash K_i\phi_1 \ &\mbox{iff }\forall t\subseteq V: s\cap O_i = t\cap O_i\mbox{ and $t$ survives in }\F^{\ell_0\ldots\ell_k}\mbox{ implies }\F^{\ell_0\ldots\ell_k},t\vDash\phi_1\\
                                                &\mbox{iff }\forall t\subseteq V: s\cap O_i = t\cap O_i\mbox{ and }\bigwedge_{i=0}^k \F^{\ell_0\ldots\ell_{i-1}},t\vDash\ell_i\mbox{ implies }\F^{\ell_0\ldots\ell_k},t\vDash\phi_1
  \end{align*}
  Note that, assuming for a $t\subseteq V$, $s\cap O_i = t\cap O_i$, $\bigwedge_{i=0}^k \F^{\ell_0\ldots\ell_{i-1}},t\vDash\ell_i$ if and only if \texttt{stillExists}$=\Tr$ after the execution of the loop in line~\ref{algoline:verifytsurv}.
  This is so by using IH as for any $i\in [1,|L|-1]$, $|\F| + |\ell_0|+\ldots+\ell_{i-1} + \ell_i < |\F| + \sum_{j=0}^k |\ell_j| + |\phi_1| + 1$. Line~\ref{algoline:checkindisting} checks whether $s\cap O_i = t\cap O_i$. Hence 
  \begin{align*}
    &\F^{\ell_0\ldots\ell_k},s\vDash K_i\phi_1\\
    \text{iff } &\forall t\subseteq V:\mbox{line~\ref{algoline:checkindisting} and \texttt{stillExists} is true after loop~\ref{algoline:verifytsurv} implies }\F^{\ell_0\ldots\ell_k},t\vDash\phi_1\\
    \text{iff } &\forall t\subseteq V:\mbox{line~\ref{algoline:checkindisting} and \texttt{stillExists} is true after loop~\ref{algoline:verifytsurv} implies }\chk(\F[\ell_0\ldots\ell_k],t,\phi_1)\neq\Fa
  \end{align*}

  \noindent\textbf{Case} $\phi = [!\psi]\phi_1$.
  We have:
  \begin{align*}
    \F^{\ell_0\ldots\ell_k},s\vDash[!\psi]\phi_1 \ &\mbox{iff }\F^{\ell_0\ldots\ell_k},s\vDash\psi\mbox{ implies }\F^{\ell_0\ldots\ell_k\psi},s\vDash\phi_1\\
                                                   &\mbox{iff }\chk(\F,L,s,\psi) = \Tr\mbox{ implies }\chk(\F,L ++ [\psi], s,\phi_1) = \Tr,\mbox{by IH}
  \end{align*}
  Note that we make two recursive calls here.
  We can apply the IH to both because the input size strictly decreases:
  $|\F| + \sum_{i=0}^k |\ell_i| + |\psi| < |\F| + \sum_{i=0}^k |\ell_i| + |\psi| + |\phi_1| + 1$ and $|\F| + \sum_{i=0}^k |\ell_i| + |\psi| + |\phi_1| < |\F| + \sum_{i=0}^k |\ell_i| + |\psi| + |\phi_1| + 1$.

  The $\chk(\F,L,s,\psi)$ is evaluated in line~\ref{algoline:pubannounceantecedent}.
  If it is $\Fa$, then $\chk(\F,[\ell_0,\ldots,\ell_k],s,[!\psi]\phi_1) = \Tr$ by line~\ref{algoline:pubannouncevacutrue}.
  Hence,
  $
  \F^{\ell_0\ldots\ell_k},s\vDash[!\psi]\phi_1\mbox{ iff }\chk(\F,[\ell_0,\ldots,\ell_1],s,[!\psi]\phi_1)=\Tr
  $.
\end{proof}

\DELcheckCorrect*
\begin{proof}
  We show the more general claim that for any list $L = [(\X_0,x_0),\ldots,(\X_k,x_k)]$ of events we have
  $(\F,s) \otimes (\X_0,x_0) \otimes \ldots \otimes (\X_k,x_k) \vDash \phi$
  iff $\chkfcdelK(\F,L,s,\phi)$ returns $\Tr$.
  The proof is by induction on $\phi$, and we omit the easy cases for negation and conjunction.
  \smallskip

  \noindent\textbf{Case} $\phi = p$.
  We have the following equivalences:
  \begin{align*}
    & \F^{\ell_0\ldots\ell_k},s\vDash p\\
    \text{iff } &p \in s^k=(s^{k-1} \setminus V_{k,-}) \cup (s^{k-1} \cap V_{k,-})^\circ \cup x^k \cup \{p \in V_{k,-} \mid s^{k-1} \cup x_k \vDash \theta_{k,-} (p)\},\mbox{by \autoref{d:trf}}\\
    \text{iff } &p \in s^{k,'}=(s^{k-1} \setminus V_{k,-}) \cup \{ q \in V_{k,-} \mid \chkfcdelK(\F, [], s, [x_j \sqsubseteq V_k ^+] \theta_{k,-}(q)) \},\mbox{because $p \in V$}\\
    \text{iff } &\chkfcdelK(\F,L,s,p) = \Tr,\mbox{by definition of checkDELK}
  \end{align*}
  Note the difference between $s^k$ and $s^{k,'}$ here. $s^{k,'}$ removes the copy variables and the event state (of which variables are from the new event vocabulary) while the last part both represent the modified propositions whose precondition was true in the previous state. However, by assumption $p \in V$, therefore we can simply remove these two parts.

  \smallskip
  \noindent\textbf{Case} $\phi = K_i \psi$.
  We have the following equivalences:
  \begin{align*}
    \ & \F^{\ell_0\ldots\ell_k},s\vDash K_i\psi\\
    \text{iff } & \mbox{for all } t \mbox{ in the final model } s^{final} \cup t' \vDash \Omega_i^{final} \mbox{ implies } \F^{final}, t \vDash \psi\\
    \text{iff } & \forall t \subseteq V \mbox{ s.t. } t \vDash \theta, t_0^+ \subseteq V_0^+,...,t_k^+ \subseteq V_k^+:\\
    & \ s \cup t'\vDash \Omega_i, x_1^+ \cup t_1^+ \vDash \Omega_{0,i}^+,..,x_k^+ \cup t_k^+ \vDash \Omega_{k_i} \mbox{ and $t$ survives implies }\F^{\ell_0\ldots\ell_k},t\vDash\psi\\
    \text{iff } & \forall t\subseteq V \mbox{ s.t. } t \vDash \theta, t_0^+ \subseteq V_0^+,...,t_k^+ \subseteq V_k^+:\\
    &\ s \cup t'\vDash \Omega_i, x_1^+ \cup t_1^+ \vDash \Omega_{0,i}^+,..,x_k^+ \cup t_k^+ \vDash \Omega_{k_i} \mbox{ and }\bigwedge_{i=0}^k \F^{\ell_0\ldots\ell_{i-1}},t\vDash [t_i^+ \sqsubseteq V_i ^+] \theta_i ^+ \mbox{ implies }\F^{\ell_0\ldots\ell_k},t\vDash \psi \\
    \text{iff } & \chkfcdelK(\F,L,s,K_i \psi) = \Tr
  \end{align*}

  \noindent\textbf{Case} $\phi = [\mathcal{X}, x] \psi$.
  We have the following equivalences:
  \begin{align*}
    &\F^{\ell_0\ldots\ell_k},s\vDash[\mathcal{X}, x]\psi\\
    \text{iff } &\mbox{if }\F^{\ell_0\ldots\ell_k},s\vDash[x \sqsubseteq V ^+] \theta ^+\mbox{ then }(\F^{\ell_0\ldots\ell_k} \times \mathcal{X},s^x) \vDash\psi\\
    \text{iff } &\mbox{if }\chkfcdelK(\F,L,s,[x \sqsubseteq V ^+] \theta ^+) = \Tr\mbox{ then }\chk(\F,L ++ [(\mathcal{X},x)], s,\psi) = \Tr,\mbox{by IH}
      \qedhere
  \end{align*}
\end{proof}

\Correcttranslationbddtomp*
\begin{proof}
By induction on the structure of $\pi$ and applying \autoref{def:mpSemantics}.
\begin{itemize}
\item $s \xrightarrow{p \leftarrow \top} t \iff s \cup t' \vDash p' \land \bigwedge_{p \neq q \in V} (q \leftrightarrow q')$, since $p \in t$ and other variables remain the same.
\item $s \xrightarrow{p \leftarrow \bot} t \iff s \cup t' \vDash \neg p' \land \bigwedge_{p \neq q \in V} (q \leftrightarrow q')$, since $p \not\in t$ and other variables remain the same.
\item $s \xrightarrow{\beta?} t \iff s = t \land s \vDash \beta \iff s \cup t' \vDash \beta \land \bigwedge_{p \in V} (p \leftrightarrow p')$, since $\beta$ must be true in $s = t$.
\item $s \xrightarrow{p \cup q} t \iff s \xrightarrow{p} t $ or $s \xrightarrow{q} t \iff s \cup t' \vDash tr(p)$ or $s \cup t' \vDash tr(q) \iff s \cup t' \vDash tr(p) \vee tr(q)$.
\item 
$\begin{array}[t]{@{}rcl}
s \xrightarrow{\pi_1;\pi_2} t &
\iff & \exists u \subseteq V: s\xrightarrow{\pi_1} u\mbox{ and } u\xrightarrow{\pi_2} t\\
& \iff & \exists u \subseteq V: s\union u'\vDash\tr(\pi_1)\mbox{ and } u\union t'\vDash\tr(\pi_2)\\
& \iff & \exists u \subseteq V: s\union u'\vDash\tr(\pi_1)\mbox{ and } u'\union t''\vDash[V\mapsto V', V'\mapsto V'']\tr(\pi_2)\\
& \iff & \exists u \subseteq V: s \union u' \union t'' \vDash\tr(\pi_1) \land [V\mapsto V', V'\mapsto V'']\tr(\pi_2)\\
& \iff & s \union t'' \vDash  \exists V' (\tr(\pi_1) \land [V\mapsto V', V'\mapsto V'']\tr(\pi_2))\\
& \iff & s \cup t' \vDash [V'' \mapsto V'] (
        \exists V' (
          \tr(\pi_1)
          \land [V\mapsto V'][V'\mapsto V'']\tr(\pi_2)
        ) 
      )
\end{array}$.
\item $s \xrightarrow{p \cap q} t \iff s \xrightarrow{p} t$ and $s \xrightarrow{q} t$
  $\iff s \cup t' \vDash tr(p)$ and $s \cup t' \vDash tr(q) \iff s \cup t' \vDash tr(p) \land tr(q)$
  \qedhere
\end{itemize}
\end{proof}

\Polyreductionmptobdd*
\begin{proof}
Take a vocabulary $V = \{p_1, \dots p_{2n}\}$ with this ordering fixed.
Consider again the formula $\beta := (p_1 \land p_{n+1} ) \lor \ldots \lor (p_n \land p_{n+n})$ from page~\pageref{largeBddForm}.
Let $\pi := \beta ?$.
Then $\pi$ also has length in $\mathcal{O}(n)$. From Definition 16, $\tr(\pi) = \beta \land \bigwedge_{p \in V} (p \leftrightarrow p')$. We show that the BDDs corresponding to this boolean formula has at least $2^{n+1}$ many nodes.
First, as discussed in~\cite[p.~681]{Bryant86}, the BDD that corresponds to the $\beta$ formula has at least $2^{n+1}$ many nodes.

\begin{center}
\begin{minipage}[b]{0.3\textwidth}
  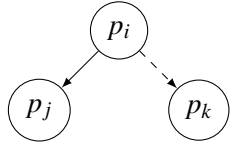
\begin{figure}[H]
    \centering
    \begin{tikzpicture}[>=latex, node distance=1.5cm]
      \node (p0) [draw,circle] {$p_i$};
      \node (p1) [draw,circle, below left of=p0] {$p_j$};
      \node (p2) [draw,circle, below right of=p0] {$p_k$};
      \draw (p0) [->,solid] -- (p1);
      \draw (p0) [->,dashed] -- (p2);
    \end{tikzpicture}
    \caption{Node $p_i$ in $\beta$'s $\mathsf{BDD}$.}\label{fig-bddofbeta}
  \end{figure}
\end{minipage}
\hspace{1cm}
\begin{minipage}[b]{0.5\textwidth}
  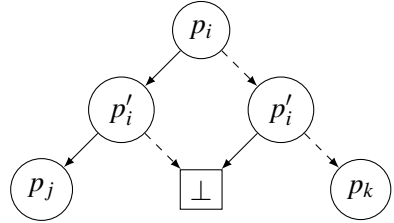
\begin{figure}[H]
    \centering
    \begin{tikzpicture}[>=latex, node distance=1.5cm]
      \node (p0) [draw,circle] {$p_i$};
      \node (p1) [draw,circle, below left of=p0] {$p_i'$};
      \node (p2) [draw,circle, below right of=p0] {$p_i'$};
      \node (p3) [draw,circle, below left of=p1] {$p_j$};
      \node (p4) [draw,circle, below right of=p2] {$p_k$};
      \node (bot) [draw,rectangle, below right of=p1] {$\bot$};
      \draw (p0) [->,solid] -- (p1);
      \draw (p0) [->,dashed] -- (p2);
      \draw (p1) [->,solid] -- (p3);
      \draw (p2) [->,dashed] -- (p4);
      \draw (p2) [->,solid] -- (bot);
      \draw (p1) [->,dashed] edge (bot);
    \end{tikzpicture}
    \caption{Node $p_i$ in $\beta \land \bigwedge_{p \in V} (p \leftrightarrow p')$'s $\mathsf{BDD}$.}\label{fig-bddofconjuncts}
  \end{figure}
\end{minipage}
\end{center}

Now we examine the size of the BDDs that corresponds to $\beta \land \bigwedge_{p \in V} (p \leftrightarrow p')$. Since none of the copy variables $p_i'$ appear in $\beta$, it will not cause any subtrees to be merged. Instead, as shown in Figure~\ref{fig-bddofconjuncts}, for each node $p_i$ that appears in $\beta's$ BDD in Figure~\ref{fig-bddofbeta}, two more nodes of $p_i'$ are added with the solid and dashed edges encoding the $(p_i \leftrightarrow p_i')$ sub-formula in $\pi$. Therefore, the BDD for $\pi$ has at least $2^{n+1}$ many nodes.
\end{proof}

\Bddtompcorrect*
\begin{proof}
For any BDD $\Omega$, let $\Omega^\ast$ be its unraveling to a tree.
Note that we have $\tau(\Omega) = \tau(\Omega^\ast)$. Hence for the proof we assume wlog.~that $\beta$ is a tree.

We proceed by induction over the structure of BDDs.
In the case when the BDD is a single node $\bot$, the mental program is $? \bot$, which represents the empty relations. When the BDD is a single node $\top$, the mental program is $((p_0 \leftarrow \bot) \union (p_0 \leftarrow \top)) ; ... ;((p_n \leftarrow \bot) \union (p_n \leftarrow \top)) ; ? \top$, which represents the complete relations. The translation is correct for both of these base cases.

The trickier case is when the BDD has at least three nodes (including $\top$ and $\bot$).
We proceed by an inner induction on the length of the vocabulary.
The base case is when the vocabulary is empty and the only BDDs that can be constructed are just $\bot$ and $\top$nodes. As shown above, the translation is correct for both cases. 

For the induction step we are adding one more variable $p_0$ at the beginning of the vocabulary $V$.
We have the inductive hypothesis that for all $t_i$ that can be constructed with the vocabulary $V=[p_1,..,p_k]$, we have $s \xrightarrow{\tau (t_0,V)} t \iff s \cup t' \vDash t_0$.

We prove that the translation remains correct for the extended vocabulary $V_1 := \{p_0\} \cup V$ by showing that both $t_0$ under $V$ and $\tau (t_0, V)$ encode the same relations.

Given the new vocabulary, the unraveled BDDs that can be constructed all have the shape as shown in Figure~\ref{fig-bddtree}, where the $t_i$s are subtrees which can be seen as the complete trees constructed in the original vocabulary $V$. We can construct the corresponding relations in Figure~\ref{fig-kripke}. Specifically, after extending the vocabulary, we make a copy of the states in the original vocabulary $V$ and add $p_0$ to each of these copies to indicate that $p_0$ is true in these states.
The states in which $p_0$ is false are put on the bottom level in the figure while the states in which $p_0$ is true are on the top level. We construct the corresponding relations as follows: 
\begin{enumerate}
    \item Construct the relations corresponding to $t_1$ in the original vocabulary $V$ on the bottom level. 
    \item Construct the relations corresponding to $t_2$ in $V$ on the bottom level and then shift the source states of the relations from the bottom level to the corresponding states (i.e.\ those that differ only by $p_0$.) on the top level. In Figure~\ref{fig-kripke} these relations are depicted as the dotted top-to-bottom arrows.
    \item Construct the relations corresponding to $t_3$ in $V$ on the bottom level and then shift the target states of the relations from the bottom level to the corresponding states on the top level. In Figure~\ref{fig-kripke} these relations are depicted as the dotted bottom-to-top arrows.
    \item Construct the relations corresponding to $t_4$ in $V$ on the bottom level and then shift both the source and target states of the relations from the bottom level to the corresponding states on the top level. 
\end{enumerate}

Each branch $t_i$ in Figure~\ref{fig-bddtree} corresponds to part of the relation in Figure~\ref{fig-kripke}.

\begin{center}
\begin{minipage}[b]{0.4\textwidth}
\begin{figure}[H]
  \centering
  \begin{tikzpicture}[>=latex, node distance=1.8cm]
    \node (p0) [draw,circle] {$p_0$};
    \node (p1) [draw,circle, below left of=p0] {$p_0'$};
    \node (p2) [draw,circle, below right of=p0] {$p_0'$};
    \node (t1) [draw,, below left of=p1] {$t_1$};
    \node (t2) [draw,, right of=t1] {$t_2$};
    \node (t4) [draw,, below right of=p2] {$t_4$};
    \node (t3) [draw,, left of=t4] {$t_3$};
    \draw (p0) [->,solid] -- (p1);
    \draw (p1) [->,dashed] -- (t2);
    \draw (p1) [->,solid] -- (t1);
    \draw (p0) [->,dashed] -- (p2);
    \draw (p2) [->,solid] -- (t3);
    \draw (p2) [->,dashed] -- (t4);
  \end{tikzpicture}
  \caption{The BDD after $p_0$ is added.}\label{fig-bddtree}
\end{figure}
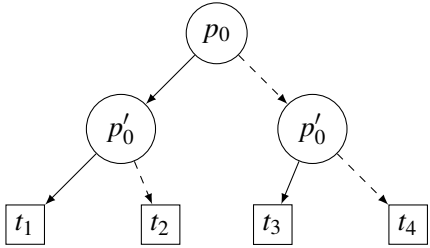
\end{minipage}%
\hspace{5mm}%
\begin{minipage}[b]{0.5\textwidth}
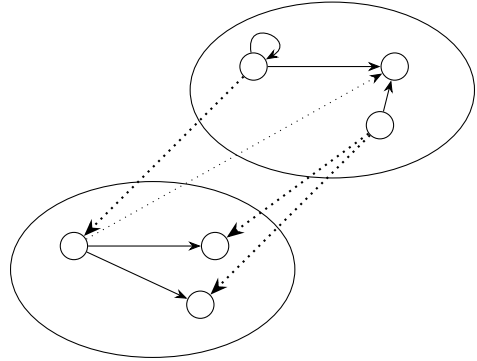
\begin{figure}[H]
  \centering
  \begin{tikzpicture}[
node distance=2cm,
main node/.style={circle, draw, minimum size=0.1cm},
>={Stealth[scale=1.0]}, scale = 0.9
]

\node[main node] (A1) {};
\node[main node, right=of A1,xshift=-0.5cm] (A2) {};
\node[main node, below right=of A1, yshift=0.9cm] (A3) {};

\node[main node, below left=3cm of A1] (B1) {};
\node[main node, right=of B1,xshift=-0.5cm] (B2) {};
\node[main node, below right=of B1, yshift=0.9cm] (B3) {};

\draw [->] (A1) to [out=100,in=30,looseness=6,->] (A1);
\draw[->] (A1) -- (A2);
\draw[->] (A3) -- (A2);

\draw[->] (B1) -- (B2);
\draw[->] (B1) -- (B3);

\draw[->, dotted, thick] (A1) -- (B1);
\draw[->, dotted, thick] (A3) -- (B2);
\draw[->, dotted, thick] (A3) -- (B3);

\draw[->, dotted, thin] (B1) -- (A2);

\coordinate (centroidA) at ($(A1)!0.333!(A2)!0.4!(A3)$);

\coordinate (centroidB) at ($(B1)!0.333!(B2)!0.4!(B3)$);

\draw (centroidA) ellipse (2.1cm and 1.3cm);
\draw (centroidB) ellipse (2.1cm and 1.3cm);

\end{tikzpicture}

  \caption{The corresponding Kripke model.}\label{fig-kripke}
\end{figure}
\end{minipage}
\end{center}

From inductive hypothesis we know that the translations are correct in the original vocabulary, which means that $\pi_i$ and $t_i$ encode the same relations in the original vocabulary $V$.
Assume $\tau(t_i,V) = \pi_i$, then $\tau(t_0,V_1)= ?\neg p_0 ; ((p_0 \leftarrow \bot;\pi_1) \cup (p_0 \leftarrow \top;\pi_2)) \cup ?p_0 ; ((p_0 \leftarrow \bot; \pi_3) \cup (p_0 \leftarrow \top;\pi_4)) = (?\neg p_0 ; \pi_1) \cup (?\neg p_0 ;(p_0 \leftarrow \top);\pi_2) \cup (?p_0;(p_0 \leftarrow \bot); \pi_3) \cup (?p_0;\pi_4)$. Similar to the reasoning in the previous part, we can see that this represents the same relation in Figure~\ref{fig-kripke}, e.g.\ the relations corresponding to $\pi_2$ go from bottom to top. Therefore, the BDD $t_0$ and the mental program $\tau(t_0, V_1)$ encode the same relation in the extended vocabulary $\{p_0\} \cup V$.
This concludes the proof that the translation is correct.
\end{proof}

\end{document}